\documentclass[conference]{IEEEtran}
\IEEEoverridecommandlockouts

\usepackage[utf8]{inputenc}
\usepackage{color}
\usepackage{amsmath,amssymb,mathtools,dsfont,bm,amsthm}
\usepackage{algorithm,algorithmic}
\usepackage{cite}
\usepackage{siunitx}

\newtheorem{remark}{Remark}
\newtheorem{lemma}{Lemma}

\usepackage{subcaption}
\usepackage{subfloat}
\IEEEsettopmargin{t}{54pt}

\DeclarePairedDelimiter{\norm}{\lVert}{\rVert}

\makeatletter
\renewenvironment{thebibliography}[1]{%
	\@xp\section\@xp*\@xp{\refname}%
	\normalfont\footnotesize\labelsep .5em\relax
	\renewcommand\theenumiv{\arabic{enumiv}}\let\p@enumiv\@empty
	\vspace*{-2pt}
	\list{\@biblabel{\theenumiv}}{\settowidth\labelwidth{\@biblabel{#1}}%
		\leftmargin\labelwidth \advance\leftmargin\labelsep
		\usecounter{enumiv}}%
	\sloppy \clubpenalty\@M \widowpenalty\clubpenalty
	\sfcode`\.=\@m
}{%
	\def\@noitemerr{\@latex@warning{Empty `thebibliography' environment}}%
	\endlist
}
\makeatother
\let\OLDthebibliography\thebibliography
\renewcommand\thebibliography[1]{
	\OLDthebibliography{#1}
	\setlength{\parskip}{-.1pt}
	\setlength{\itemsep}{-.5pt}
}
\title{Energy Efficiency in Rate-Splitting Multiple Access with Mixed Criticality}
\author{\IEEEauthorblockN{Robert-Jeron Reifert, Stefan Roth, Alaa Alameer Ahmad, and Aydin Sezgin}
\IEEEauthorblockA{Institute of Digital Communication Systems, Ruhr University Bochum, Bochum, Germany\\
Email: \{robert-.reifert,stefan.roth-k21,alaa.alameerahmad,aydin.sezgin\}@rub.de}
\thanks{This work was funded in part by the Federal Ministry of Education and Research (BMBF) of the Federal Republic of Germany (F\"orderkennzeichen 01IS18063A, ReMiX), and in part by the Deutsche Forschungsgemeinschaft (DFG, German Research Foundation) under Germany's Excellence Strategy - EXC 2092 CASA - 390781972.\newline
This work has been submitted to the IEEE for possible publication. Copyright may be transferred without notice, after which
this version may no longer be accessible.}
}
\date{December 2021}

\begin{document}

\maketitle\vspace*{-.6cm}
\begin{abstract}
	Future sixth generation (6G) wireless communication networks face the need to similarly meet unprecedented quality of service (QoS) demands while also providing a larger energy efficiency (EE) to minimize their carbon footprint. Moreover, due to the diverseness of network participants, mixed criticality QoS levels are assigned to the users of such networks. 
	%
	%
	%
	In this work, with a focus on a cloud-radio access network (C-RAN), the fulfillment of desired QoS and minimized transmit power use is optimized jointly within a rate-splitting paradigm.
	%
	%
	%
	%
	Thereby, the optimization problem is non-convex. Hence, a low-complexity algorithm is proposed based on fractional programming.
	%
	%
	Numerical results validate that there is a trade-off between the QoS fulfillment and power minimization. Moreover, the energy efficiency of the proposed rate-splitting algorithm is larger than in comparative schemes, especially with mixed criticality.
	%
\end{abstract}
	\vspace*{-.15cm}
\section{Introduction}\vspace*{-.2cm}
The road towards the sixth generation (6G) of wireless communication networks is already being pursued by researchers around the globe. Through a wide field of applications, the empowerment of anytime anywhere access, and an overwhelming amount of connected devices, 6G brings enormous challenges towards the development of future network technologies. To ensure carbon neutrality, 6G networks are expected to be green whilst also fulfilling quality of service (QoS) requirements in an energy efficient manner \cite{8922617}. \\
\indent To ensure energy efficiency (EE) under fulfilling the QoS demands, we investigate QoS target capabilities within a cloud-radio access network (C-RAN), in which various users are connected to multiple base stations (BSs), which are jointly controlled by a central processor (CP) at the cloud, as drawn in \figurename~\ref{sys_mdl}. A C-RAN is a promising network architecture, which enables centralization and virtualization providing high elasticity, high QoS, and good EE \cite{7378422}. Thereby, the QoS assigned to the users are designed to match the desired data rates of the users (target rates), which themselves often depend on the subscribed contract. Hence, we aim at designing the C-RAN to enable \emph{mixed criticality} regarding different communication links. In industrial context, mixed criticality corresponds to different importance of network participants, e.g., a security monitoring system is more critical than a maintenance scheduler. Since the introduction of mixed criticality in 2007 \cite{sota17} for task scheduling in real-time systems, many research works discussed, reviewed, and analyzed mixed criticality in communications, e.g., works \cite{sota18,sota21}.\\
\indent For the communication between the the BSs and the users, we are employing \textit{rate-splitting}. Originating in the early 80's \cite{1056307}, and shown to achieve within one-bit of the interference channel capacity \cite{4675741}, rate-splitting mutliple access (RSMA) achieved significant attention in research, e.g., \cite{5910112,8846706,mao2022ratesplitting}. Contrary to only considering single message streams to each user (\emph{private message}), under RSMA, a \emph{common message} is utilized for two reasons: $(1)$ The common message of a user is used analog to the private message, to transmit additional data to the user. $(2)$ For the purpose of interference mitigation, common message decoding at other users is utilized. 
That is, a user may decode a subset of all common messages, including its own, in a successive manner, in order to reduce the interference level when decoding its own private message. \\
    \setlength{\textfloatsep}{0pt}%
	\begin{figure}
		\centering
		\includegraphics[scale=0.82]{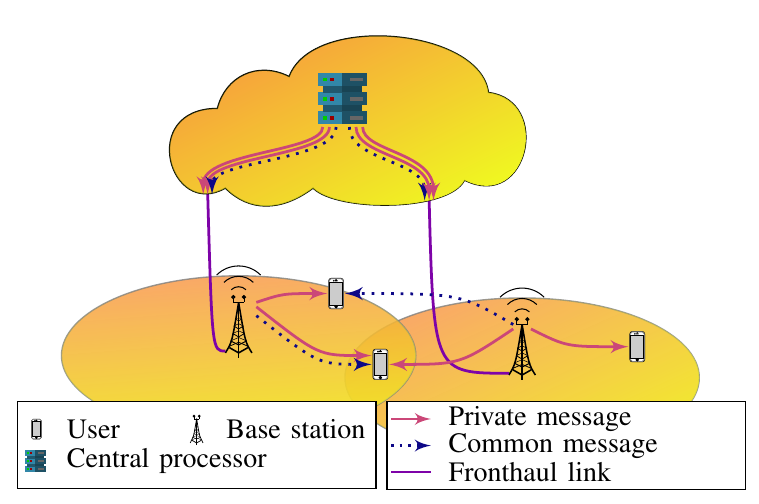}
		\vspace*{-.3cm}
		\caption{System model of a C-RAN consisting of $2$ BSs and $3$ users, where private and common messages are transmitted.}
		\label{sys_mdl}
		\vspace*{-.05cm}
	\end{figure}%
\indent There are recent related works considering the EE of the RSMA paradigm \cite{8846706,8491100,9650662,9145076}, see also references therein. Work \cite{8846706} considers layered-based RSMA techniques under weighted sum rate and EE objective subject to QoS constraints. Simulations validated the enhancement in spectral and energy efficiency under non-orthogonal unicast and multicast transmission. Similarly, \cite{8491100} and \cite{9650662} asess the EE of RSMA compared to space-division and non-orthogonal multiple access, denoted by SDMA and NOMA, respectively. These works, however, do not reside in the C-RAN domain, nor do they consider QoS demands. Especially, \cite{9650662} studies the trade-off beteen spectral and energy efficiency as two conflicting objectives. In \cite{9145076}, the EE of RSMA was considered under C-RAN architecture comparing numerically efficient and global optimal approaches.

Moreover, \cite{9145363} investigates power minimization under a QoS constraint. Therein, the authors investigate the minimization of weighted-sum of transmit powers subject to per-user QoS constraints. Hence, such scheme is only working in networks, where the QoS is achievable, same goes for \cite{8846706}. However, this assumption is rather optimistic and we herein propose a more generalized scheme.

In this work, we consider the joint minimization of transmit power and mean squared error (MSE) of QoS deviation, i.e., the gap of allocated and desired rate. Thereby, we utilize a mixed critical C-RAN under the RSMA paradigm in order to achieve good EE whilst fulfilling QoS demands. As such, we jointly optimize the precoding vectors and allocated rates subject to per-BS fronthaul capacity, maximum transmit power, and per-user achievable rate constraints.
\section{System Model and Optimization Metric}\label{sec:sysmod}
    \begin{table}[t]\vspace*{.15cm}
	\caption{List of Network Parameters}\vspace*{-.3cm}
	\label{tb:notations}
	\centering
	\begin{tabular}{|c|l|}
		\hline
		Notation & Definition\\
		\hline
		\hline
		$\mathcal{B}$ & Set of BSs\\
		$\mathcal{K}$ & Set of single-antenna users\\
		${L}$ & Number of antennas per BS\\
		${C}_b^{\text{max}}$ & Fronthaul capacity per BS\\
		${P}_b^{\text{max}}$ & Maximum transmit power per BS\\
		$\tau$ & Transmission bandwidth in MHz\\
		\hline
	\end{tabular}
\end{table}%
	The network considered is a downlink C-RAN utilizing \emph{data-sharing} transmission strategy. Under such architecture, a cloud coordinates ${B}$ BSs via fronthaul links in order to serve ${K}$ users, where the CP at the cloud performs most baseband processing tasks. That is, the CP encodes messages into signals and designs the joint precoding vectors, which are then forwarded to the BSs to perform radio transmission. For easy access, Table \ref{tb:notations} lists all relevant system and network parameters. We assume the cloud to have access to the full channel state information (CSIT), which is reasoned in the assumption of a \emph{block-based transmission model}. A transmission block is made of a couple of time slots in which the channel state remains constant, thus the CSIT needs to be acquired at the beginning of each block. The proposed algorithm optimizes the resource allocation within one such block. 
	
	In this system, messages are coded via the RSMA framework. The requested content of user $k$ will be split into a private and common message. Thereby, the CP thus encodes the messages into $s_k^p$ and $s_k^c$, the private and common signal to be transmitted to user $k$. These signals are zero mean, unit variance complex Gaussian variables with the property of being independent identically distributed and circularly symmetric. While $s_k^p$ is intended to be decoded by user $k$ only, RSMA employs $s_k^c$ signals to be decoded at multiple users for the purpose of interference mitigation. This necessitates the utilization of a successive decoding strategy at the users.
	
	Now, let $\bm{h}_{b,k}\in \mathbb{C}^{L\times 1}$ be the channel vector linking user $k$ and BS $b$. Thus, the aggregate channel vector of $k$ is $\bm{h}_{k}=[(\bm{h}_{1,k})^T,\ldots,(\bm{h}_{B,k})^T]^T\in\mathbb{C}^{L B\times 1}$. The aggregate precoding vectors for transmitting $k$'s signals are $\bm{w}_{k}^o=[(\bm{w}_{1,k}^o)^T,\ldots,(\bm{w}_{B,k}^o)^T]^T\in\mathbb{C}^{L B\times 1}$, which consist of the individual precoding vectors $\bm{w}_{b,k}^o\in\mathbb{C}^{L\times 1}$, where $o\in\{p,c\}$ denotes private and common vectors, respectively. Throughout this work, the index $o$ denotes the differentiation of private and common signal related variables. Due to limited radio resources, BSs naturally have limited capabilities regarding the number of served users. Hence, we introduce the sets $\mathcal{K}_b^p$ and $\mathcal{K}_b^c$, which include only the users whose private or common signal is served by BS $b$. Hence, the previously defined precoding vectors often contain zeros, i.e., $\bm{w}_{b,k}^o = \bm{0}_L$ when $k\notin\mathcal{K}_b^o$. Note that in this work we assume the clustering to be fixed by \cite[Algorithm 1]{9445019}. For cases, where an additional optimization of the clustering is needed, we refer to \cite{9217249}.\\
	\indent Each message stream transmits the data via a specific rate $r_k^o$, while the total rate assigned to user $k$ is $r_k = r_k^p + r_k^c$. To ensure operation of the considered network, the CP has to respect the finite fronthaul capacity of the CP-BS links with
	\begin{equation}\label{eq:fthl}
	    \sideset{}{_{k\in\mathcal{K}_b^p}}\sum r_k^p + \sideset{}{_{k\in\mathcal{K}_b^c}}\sum r_k^c \leq {C}_b^\mathrm{max}.
	\end{equation}
	In what follows, we explain the construction of the transmit signal and the successive decoding scheme for the common streams.
	
	\subsection{Transmit Signal and Successive Decoding}
	Each BS constructs its transmit signal $\bm{x}_b \in\mathbb{C}^{L\times 1}$ according to the precoder coefficients and the user signals by calculating
	\begin{align}
	    \bm{x}_b = \sideset{}{_{k\in\mathcal{K}_b^p}}\sum \bm{w}_{b,k}^p s_k^p + \sideset{}{_{k\in\mathcal{K}_b^c}}\sum \bm{w}_{b,k}^c s_k^c.\label{eq:x_n_def}
	\end{align}
	Thereby, the signal transmitted by each BS is subject to the maximum transmit power constraint $\mathbb{E}\{ \bm{x}_b^H \bm{x}_b \}$ as
	\begin{equation}\label{eq:pwr}
	    \sideset{}{_{k\in\mathcal{K}_b^p}}\sum \norm[\big]{\bm{w}_{b,k}^p}_2^2 + \sideset{}{_{k\in\mathcal{K}_b^c}}\sum \norm[\big]{\bm{w}_{b,k}^c}_2^2 \leq P_b^\mathrm{max}.
	\end{equation}
	
	
	In the RSMA framework, multiple users may decode each common message to reduce the interference for messages decoded afterwards. It is thus relevant to consider additional definitions of the network, which are provided as follows:
	
	\begin{itemize}
	    \item The set of users, which decode user $k$'s common message, is $\mathcal{M}_k = \{j\in\mathcal{K}| \text{user } j \text{ decodes } s_k^c\}$.
	    \item The users, whose common messages are decoded by user $k$, are denoted in the set $\mathcal{I}_k = \{j\in\mathcal{K}| k\in \mathcal{M}_j \}$.
	    \item The decoding order at user $k$ is written as $\pi_k$, where $\pi_k(m) > \pi_k(i)$ means that user $k$ decodes common message $i$ before message $m$.
	    \item The set of users, whose common messages are decoded after decoding user $i$'s message at user $k$, become $\mathcal{I}'_{i,k} = \{ m\in\mathcal{I}_k | \pi_k(m) > \pi_k(i) \}$.
	\end{itemize}
	A suitable method of calculating $\mathcal{M}_k$, $\mathcal{I}_k$, $\mathcal{I}'_{i,k}$, and $\pi_k$ is provided by \cite{9445019}.
	
	Taking those definitions into account, the received signal at user $k$ can be formulated as
	\begin{align}\label{eq:y_k}
	    &y_k = \bm{h}_{k}^H\bm{w}_{k}^p {s}_{k}^p + \sum_{j \in \mathcal{I}_{k}}\bm{h}_{k}^H\bm{w}_{j}^c {s}_{j}^c + \sum_{m \in \mathcal{K}\setminus \{k\}}\bm{h}_{k}^H\bm{w}_{m}^p {s}_{m}^p \nonumber\\
	    &\qquad\qquad + \sum_{l \in \mathcal{K}\setminus \mathcal{I}_k}\bm{h}_{k}^H\bm{w}_{l}^c {s}_{l}^c+ n_k .
	\end{align}
	Here, $n_k \sim \mathcal{C}\mathcal{N}(0,\sigma^2)$ represents additive white Gaussian noise, assumed to have the same power for all users. Hence, \eqref{eq:y_k} contains two terms of signals to be decoded during the successive decoding, namely the first two terms consisting of private and common signals. The last three terms in \eqref{eq:y_k} denote interference from private signals, common signals, and noise, respectively. Using this definition, the signal to interference plus noise ratios (SINRs) of the messages decoded are \vspace*{-0.2cm}
	\begin{subequations}
	\begingroup
	\addtolength{\jot}{-.1cm}
    \begin{align}
		\label{eq:e2.18}
		\Gamma_{k}^p &= \frac{\left|\bm{h}_{k}^{H}\bm{w}_{k}^p \right|^2}{\sum\limits_{j \in \mathcal{K}\setminus \{k\}}\left|\bm{h}_{k}^{H}\bm{w}_{j}^p \right|^2 + \sum\limits_{l \in \mathcal{K}\setminus \mathcal{I}_k}\left|\bm{h}_{k}^{H}\bm{w}_{l}^c \right|^2+\sigma^2},\\
		\label{eq:e2.19}
		\Gamma_{i,k}^c &=\frac{\left|\bm{h}_{k}^{H}\bm{w}_{i}^c \right|^2}{\sum\limits_{j \in \mathcal{K}}\left|\bm{h}_{k}^{H}\bm{w}_{j}^p \right|^2 + \sum\limits_{l \in \mathcal{K}\setminus \mathcal{I}_k\cup \mathcal{I}'_{k,i}}\left|\bm{h}_{k}^{H}\bm{w}_{l}^c \right|^2+\sigma^2 }.
	\end{align}\endgroup
	\end{subequations}
	\\ \vspace*{-.8cm} \\
	Here, $\Gamma_{k}^p$ is the SINR of the private message and $\Gamma_{i,k}^c$ the SINR of user $i$'s common message, both decoded at user $k$.
	\subsection{Considered MSE and EE Metric}
	In this work, we are analyzing the merits of two system properties, i.e., we want to minimize the gap of desired and allocated rate (i.e., MSE of QoS deviation) together with the energy consumption. Therefore, we optimize a metric covering a weighted sum of both targets, i.e.,\vspace*{-0.2cm}
	\begingroup
	\addtolength{\jot}{-.15cm}
	\begin{align}\label{eq:f}
	    &\Psi = \alpha\frac{1}{|\mathcal{K}|}\sum_{k\in\mathcal{K}}\left|\left(r_k^p + r_k^c\right) - r^{\mathrm{des}}_k\right|^2 \nonumber\\ 
	    &\qquad\qquad+ (1-\alpha)\sum_{k\in\mathcal{K}} \left(\norm[\big]{\bm{w}_k^p}_2^2 + \norm[\big]{\bm{w}_k^c}_2^2\right).
	\end{align}\endgroup
	\\ \vspace*{-.8cm}\\
	Thereby, the mixed criticality of links is represented in $r^{\mathrm{des}}_k$, where critical applications have greater QoS requirements than others. Parameter $\alpha\in[0,1]$ denotes a weighting factor to achieve custom trade-offs between rate gap and power minimization.
	Note that the transmit power covered by \eqref{eq:f} only represents parts of the energy consumption of the network. 
	A more general power metric would consider operating power, fronthaul link power, and processing power. However, as we do not model or optimize power related variables other than precoding vectors, such model is reasonable.
	\begin{remark}
	Minimizing a function based on the rate gap and power usage as in \eqref{eq:f} contributes to finding a good trade-off between minimizing the power usage and maximizing the data rates. On the one hand, systems which can hardly fulfill the desired QoS fall into a mode of maximizing each user's rate upon meeting the demands. On the other hand, as rate targets are met, no further enhancement of the rates is necessary, prioritizing the problem of minimizing transmit power.
	\end{remark}\vspace*{-0.2cm}
	
	\section{Problem Formulation and Algorithm}\label{sec:opt}\vspace*{-0.1cm}
	The problem under consideration is a joint minimization of the network-wide rate gap and the transmit power consumption, which is formulated as follows:\vspace*{-0.2cm}
	\begin{subequations}\label{eq:Opt1}
	\begingroup
	\addtolength{\jot}{-.1cm}
		\begin{align}
			\vspace*{-1cm}\underset{\bm{w},\bm{r}}{\mathrm{min}}\quad &\Psi  \tag{\ref{eq:Opt1}} \\
			\mathrm{s.t.} \quad &\eqref{eq:fthl},\ \eqref{eq:pwr}, \nonumber\\
			& r_{k}^{p} \leq \tau\,\log_2(1+\Gamma_{k}^p),  &\forall k &\in \mathcal{K}, \label{eq:achp}\\	
			& r_{i}^{c} \leq \tau\,\log_2(1+\Gamma_{i,k}^c), &\forall i\in\mathcal{I}_k, \forall k &\in \mathcal{K}. \label{eq:achc}
		\end{align}\endgroup
	\end{subequations} 
	Problem \eqref{eq:Opt1} minimizes the MSE of assigned (private and common) rate to the desired rate and the transmit power by jointly managing the allocated rates and precoding vectors. Hereby, the precoding and rate vector \vspace*{-0.15cm}
	\begin{align*}
	    \bm{w}&=[(\bm{w}_{1}^p)^T,\ldots,(\bm{w}_{K}^p)^T,(\bm{w}_{1}^c)^T,\ldots,(\bm{w}_{K}^c)^T]^T, \\
	    \bm{r}&=[r_{1}^p,\ldots,r_{K}^p,r_{1}^c,\ldots,r_{K}^c]^T,
	\end{align*}
	\\\vspace*{-0.95cm}\\
	denote the optimization variables, respectively. The feasible set of problem \eqref{eq:Opt1} is defined by the fronthaul capacity \eqref{eq:fthl}, the maximum transmit power \eqref{eq:pwr}, and the achievable rates \eqref{eq:achp} and \eqref{eq:achc}. Due to the dependence on the SINR variables in \eqref{eq:e2.18} and \eqref{eq:e2.19}, the constraints \eqref{eq:achp} and \eqref{eq:achc} are non-conxex.	However, an objective function such as \eqref{eq:Opt1} is already in convex form, since the precoder norms and the MSE are both convex. 
	
	A few notes on the feasible set of problem \eqref{eq:Opt1}. The constraints \eqref{eq:fthl} and \eqref{eq:pwr} are highly dependent on the clustering sets $\mathcal{K}_b^p$ and $\mathcal{K}_b^c$, which are herein assumed to be fixed a priori. Constraint \eqref{eq:achc} takes the form of a multiple access constraint. That is, $r_i^c$ is the rate of user $i$'s common stream and the subset of users $\mathcal{M}_i$ are going to decode that specific signal. Thereby, the lowest $\Gamma_{i,k}^c$ determines $r_i^c$, i.e., the lowest SINR of all decoding users bounds the rate. Next, we apply a technique for convexifying the non-convex constraints \eqref{eq:achp} and \eqref{eq:achc}, which is based on fractional programming.
	
	\subsection{Quadratic Transform-based Solution}
	The non-convexity of problem \eqref{eq:Opt1} stems from the constraints \eqref{eq:achp} and \eqref{eq:achc}, as these have a complex fractional form. By introducing the auxiliary variables $\gamma_k^p$ and $\gamma_{i,k}^c$ for the SINR, we transform \eqref{eq:Opt1} into the optimization problem presented in Lemma~\ref{lma_1}. \vspace*{-0.15cm}
	\begin{lemma}\label{lma_1}
	    The optimization problem \eqref{eq:Opt1} can be rewritten as\vspace*{-0.15cm}
    	\begin{subequations}\label{eq:Opt2}
        \begingroup
        \addtolength{\jot}{-.1cm}
    		\begin{align}
    			\underset{\bm{w},\bm{r},\bm{\gamma}}{\mathrm{min}}\quad &\Psi  \tag{\ref{eq:Opt2}} \\
    			\mathrm{s.t.} \quad &\eqref{eq:fthl},\ \eqref{eq:pwr}, \nonumber\\
	    & r_{k}^{p} \leq \tau\,\log_2(1+\gamma_{k}^p),  &\forall k &\in \mathcal{K}, \label{eq:achp2}\\	
		& r_{i}^{c} \leq \tau\,\log_2(1+\gamma_{i,k}^c), &\forall i\in\mathcal{I}_k, \forall k &\in \mathcal{K}, \label{eq:achc2}\\
	    & \gamma_{k}^p \leq \Gamma_{k}^p,  &\forall k &\in \mathcal{K}, \label{eq:sinrp}\\	
		& \gamma_{i,k}^c \leq \Gamma_{i,k}^c, &\forall i\in\mathcal{I}_k, \forall k &\in \mathcal{K}. \label{eq:sinrc}
    		\end{align}\endgroup
    	\end{subequations} 
    	\\ \vspace*{-1cm} \\
	    Thereby, the introduced optimization variable $\bm{\gamma}$ covers all possible non-zero elements of $\bm{\gamma}'=[{\gamma}_{1}^p,\ldots,{\gamma}_{K}^p,{\gamma}_{1,1}^c,{\gamma}_{1,2}^c,\ldots,{\gamma}_{K,K}^c]^T$. Note that ${\gamma}^c_{i,k} = 0$ $\forall i\notin\mathcal{I}_k, k \in \mathcal{K}$ as only a part of all users decode common message $i$.
	    With $(\bm{w}^\star,\bm{r}^\star,\bm{\gamma}^\star)$ being a stationary solution to problem \eqref{eq:Opt2}, $(\bm{w}^\star,\bm{r}^\star)$ is a stationary solution of \eqref{eq:Opt1}.
	\end{lemma}\vspace*{-0.15cm}
	\begin{proof}
	    Both problems \eqref{eq:Opt1} and \eqref{eq:Opt2} share the same objective and 
	    constraints \eqref{eq:achp2} and \eqref{eq:sinrp} can be formulated as \vspace{-.15cm}
	    \begin{equation}
	        r_{k}^{p} \leq B\,\log_2(1+\gamma_{k}^p) \leq B\,\log_2(1+\Gamma_{k}^p).\vspace{-.15cm}
	    \end{equation}
	    When engineering the newly introduced optimization variable $\gamma_k^p$ to fulfill $r_{k}^{p} = B\,\log_2(1+\gamma_{k}^p)$ and $\gamma_{k}^p = \Gamma_{k}^p$, we get \eqref{eq:achp}. A similar statement holds for the common messages, which completes the proof.
	\end{proof}
	In Lemma~\ref{lma_1}, the constraints \eqref{eq:achp2} and \eqref{eq:achc2} appear in convex form, however, \eqref{eq:sinrp} and \eqref{eq:sinrc} are still non-convex. Let us now tackle these constraints via fractional programming. 
	That is, after subtracting the right term from both sides, we utilize the quadratic transform (QT) in multidimensional and complex case proposed by \cite[Theorem 2]{8314727} to reformulate the fraction of signal to interference plus noise. Thereby, we obtain the following functions\vspace*{-0.15cm}
	\begin{align}
		&g^p(\bm{w},\bm{\gamma}) = \gamma_{k}^p - 2 \mathrm{Re}\left\{ \left({u}_{k}^p\right)^* \left(\bm{w}_{k}^p\right)^H \bm{h}_{k} \right\} \label{eq:qtp}\\
		&+ |u_{k}^p|^2 \left[ \sigma^2 + \sum\limits_{j \in \mathcal{K}\setminus \{k\}}\left|\bm{h}_{k}^{H}\bm{w}_{j}^p \right|^2 + \sum\limits_{l \in \mathcal{K}\setminus \mathcal{I}_k}\left|\bm{h}_{k}^{H}\bm{w}_{l}^c \right|^2 \right],\hspace{-0.1cm} \nonumber\\
		&g^c(\bm{w},\bm{\gamma}) = \gamma_{i,k}^c - 2 \mathrm{Re}\left\{ \left({u}_{i,k}^c\right)^* \left(\bm{w}_{i}^c\right)^H \bm{h}_{k} \right\} \label{eq:qtc}\\
		&+ |u_{i,k}^c|^2 \left[ \sigma^2 + \sum\limits_{j \in \mathcal{K}}\left|\bm{h}_{k}^{H}\bm{w}_{j}^p \right|^2 + \sum\limits_{l \in \mathcal{K}\setminus \mathcal{I}_k\cup \mathcal{I}'_{k,i}}\left|\bm{h}_{k}^{H}\bm{w}_{l}^c \right|^2 \right].\hspace{-0.1cm} \nonumber
	\end{align}
	Here, \eqref{eq:qtp} and \eqref{eq:qtc} are the QT associates of \eqref{eq:sinrp} and \eqref{eq:sinrc}, respectively, with ${u}_{k}^p$ and ${u}_{i,k}^c$ being auxiliary variables.
	\begin{remark}
	    For fixed ${u}_{k}^p$ and ${u}_{i,k}^c$, the second terms in \eqref{eq:qtp} and \eqref{eq:qtc} become linear functions of the precoders, also the latter terms become convex. Thus, \eqref{eq:qtp} and \eqref{eq:qtc} denote convex functions of the procoding vectors and the SINR variables, in case the auxiliary variables are fixed.
	\end{remark}
	To now obtain optimal auxiliary variables for fixed $\bm{w}$ and $\bm{\gamma}$, we consider Lemma~\ref{lem:u}. 
	\begin{lemma}\label{lem:u}
	    The optimal auxiliary variable results are
	\begin{align}
	    {u}_{k}^p = \frac{\left(\bm{w}_{k}^p\right)^H \bm{h}_{k}}{\sigma^2 + \sum\limits_{j \in \mathcal{K}\setminus \{k\}}\left|\bm{h}_{k}^{H}\bm{w}_{j}^p \right|^2 + \sum\limits_{l \in \mathcal{K}\setminus \mathcal{I}_k}\left|\bm{h}_{k}^{H}\bm{w}_{l}^c \right|^2},\label{eq:ukp}\\
	    {u}_{i,k}^c = \frac{\left(\bm{w}_{i}^c\right)^H \bm{h}_{k}}{\sigma^2 + \sum\limits_{j \in \mathcal{K}}\left|\bm{h}_{k}^{H}\bm{w}_{j}^p \right|^2 + \sum\limits_{l \in \mathcal{K}\setminus \mathcal{I}_k\cup \mathcal{I}'_{k,i}}\left|\bm{h}_{k}^{H}\bm{w}_{l}^c \right|^2}.\label{eq:ukc}
	\end{align}
	\end{lemma}
	\begin{proof}
	    In \eqref{eq:qtp} and \eqref{eq:qtc}, the partial derivatives of $g^p(\bm{w},\bm{\gamma})$ and $g^c(\bm{w},\bm{\gamma})$ with respect to ${u}_{k}^p$ and ${u}_{i,k}^c$ are set to zero and solved, respectively. From this follow the optimal auxiliary variable results in \eqref{eq:ukp} and \eqref{eq:ukc}.
	\end{proof}
	As the optimal auxiliary variables depend on the selected precoders, the procedure requires an iterative solution. 
	At each iteration, the optimization problem becomes
	\begin{subequations}\label{eq:Opt3}
		\begin{align}
			\underset{\bm{w},\bm{r},{\bm \gamma}}{\mathrm{min}}\quad &\Psi  \tag{\ref{eq:Opt3}} \\
			\mathrm{s.t.} \quad &\eqref{eq:fthl},\ \eqref{eq:pwr},\ \eqref{eq:achp2},\ \eqref{eq:achc2}, \nonumber\\
			&g^p(\bm{w},\bm{\gamma}) \leq 0,  &\forall k&\in \mathcal{K}, \label{eq:gpqt}\\
			&g^c(\bm{w},\bm{\gamma}) \leq 0, &\forall i\in\mathcal{M}_k, \forall k&\in \mathcal{K}. \label{eq:gcqt}
		\end{align}
	\end{subequations} 
	Here the objective function \eqref{eq:Opt3} and the feasible set defined by all constraints are convex, problem \eqref{eq:Opt3} is a convex optimization problem that can efficiently solved by established solvers, such as CVX \cite{cvx}.
	
	To solve problem \eqref{eq:Opt3} efficiently, we propose detailed steps for precoder design and rate allocation in Algorithm \ref{alg}. At the beginning, we need to initiate the predoding vectors with feasible values. This can be done by \emph{random} initialization or by computing the maximum ratio transmitters. Then, we fix $\mathcal{M}_k$, $\mathcal{I}_k$, $\mathcal{I}'_{i,k}$, and the decoding order $\pi_k$, as well as the clustering sets for private and common streams $\mathcal{K}_b^p$ and $\mathcal{K}_b^p$. 
	Afterwards, the iterative algorithm repeats the following two steps until convergence: $(a)$ First, the auxiliary variables are updated according to \eqref{eq:ukp} and \eqref{eq:ukc}; $(b)$ Secondly, problem \eqref{eq:Opt3} is solved using CVX.
	\begin{subfigures}
    \begin{figure}
    \vspace{-.3cm}
    \end{figure}
	\begin{algorithm}[t]
    \caption{Resource Management for MSE and EE under RSMA}
    \begin{algorithmic}[1]
	\STATE Initialize feasible precoders $\bm{w}$\vspace*{-.05cm}
    \STATE Determine RSMA-related sets $\mathcal{M}_k$, $\mathcal{I}_k$, $\mathcal{I}'_{i,k}$, $\pi_k$\vspace*{-.05cm}
    \STATE Determine the clustering $\mathcal{K}_b^p$, $\mathcal{K}_b^p$\\\vspace*{-.05cm}
	\REPEAT
	\STATE Update the ${u}_{k}^p$ and ${u}_{i,k}^c$ using \eqref{eq:ukp} and \eqref{eq:ukc}\vspace*{-.05cm}
	\STATE Solve the convex problem \eqref{eq:Opt3} 
	\UNTIL{convergence}
    \end{algorithmic}
    \label{alg}
    \end{algorithm}%
    \end{subfigures}
    
    \begin{lemma}
        The iterative procedure in Algorithm \ref{alg} yields a stationary solution $(\bm{w}^\star,\bm{r}^\star)$ to problem \eqref{eq:Opt1}.
    \end{lemma}
    \begin{proof}
        Please refer to Appendix \ref{app2}.
    \end{proof}
    
    We note that the complexity of Algorithm \ref{alg} depends on the convergence rate and the complexity of problem \eqref{eq:Opt3}. Due to the form of objective and variables, the problem can be cast as a second order cone program (SOCP), which is solvable using interior-point methods. The overall upper-bound computational complexity of Algorithm \ref{alg} is $\mathcal{O}({V}_{\text{max}}(d_1)^{3.5})$, where $d_1={K}(2({B}{L}+1)+{K}+1)$ is the total number of variables and ${V}_{\text{max}}$ is the number of iterations until convergence in the worst-case.

    \section{Numerical Simulation}\label{sec:sim}
	To evaluate the performance of the proposed methods, we conduct numerical simulations in this section. Therefore, we consider a network over a square area of $\SI{800}{\meter}$ by $\SI{800}{\meter}$, in which BSs and users are placed randomly. Each BS is equipped with $L=2$ antennas and has a maximum transmit power of $P_n^\mathrm{max} = \SI{28}{dBm}$. Unless mentioned otherwise, we fix $\alpha = 0.5$, $C_n^\mathrm{max}=\SI{28}{Mbps}$, the number of BS $N=10$, and the number of uses $K=16$. 
	We consider a channel bandwidth of $B = \SI{10}{MHz}$. In the considered channel model, the path-loss is modeled by $\mathrm{PL}_{n,k}= 128.1 + 37.6\cdot\mathrm{log}_{10}(d_{n,k})$, where the distance of user $k$ and BS $n$ is denoted as $d_{n,k}$. Additionally, we consider log-normal shadowing with $\SI{8}{dB}$ standard deviation and Rayleigh fading with zero mean and unit variance. The noise power spectral density is $\SI{168}{dBm/Hz}$. Unless mentioned otherwise, we set the mixed-critical QoS demands to $\SI{14}{Mbps}$, $\SI{7}{Mbps}$, and $\SI{3}{Mbps}$, for $4$, $6$, and $6$, random users, respectively. Thereby, we employ three criticality levels, namely high (HI), medium (ME), and low (LO).
	
	In this work, we have proposed a specific method, which we refer to as RSMA. Additionally, we consider $2$ different reference schemes, treating interference as noise (TIN), and a single common message-based (SCM-RSMA) scheme. 
	TIN does not consider any rate-splitting capabilities. In contrast, SCM-RSMA employs a one-layer rate-splitting, where one \emph{super common} stream is decoded by all users in addition to private messages, e.g., used in \cite{8846706,7513415}. 
	\subsection{Impact of Interference Management Schemes}
	\begin{figure}[t]
	\centering
	\begin{subfigure}[t]{0.24\textwidth}
		\begin{center}
			\includegraphics[scale=0.59]{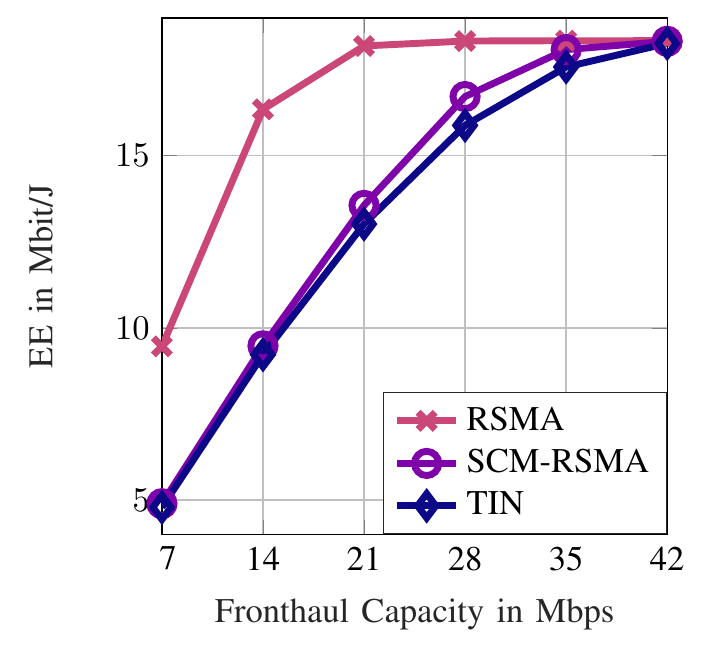}\vspace{-.2cm}
			\caption{EE vs. ${C}_b^\mathrm{max}$.}
			\label{xEE_yFthl_v3}
		\end{center}
	\end{subfigure}\hfill
	\begin{subfigure}[t]{0.24\textwidth}
		\centering
		\includegraphics[scale=0.59]{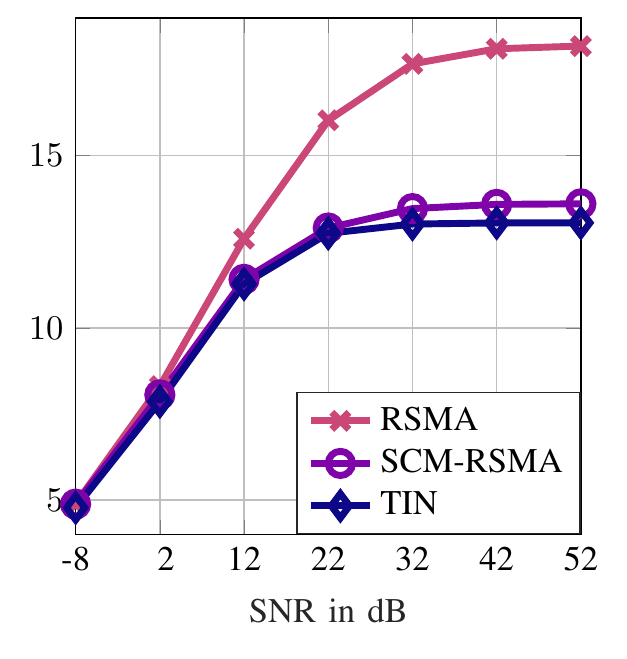}\vspace{-.2cm}
		\caption{EE vs. SNR.}
		\label{yEE_xSNR_v1}
	\end{subfigure}
	\caption{EE as a function of fronthaul capacity and SNR for the considered schemes.} \label{bla2}
    \end{figure}
	In the first set of simulations, we compare the impact of different interference management schemes, i.e., RSMA, SCM-RSMA, and TIN, on the EE of the system. 	As the EE is of special interest for future 6G networks, we herein consider such metric as follows 
    \begin{align}
	    &\Phi = \frac{\sum_{k\in\mathcal{K}} r_k}{\sum_{k\in\mathcal{K}} \left(\norm[\big]{\bm{w}_k^p}_2^2 + \norm[\big]{\bm{w}_k^c}_2^2\right) + P^{\mathrm{cir}}}.\label{eq:phi}
	\end{align}
	Here $P^{\mathrm{cir}}$ is a fixed power value associated to the operational power consumption of the C-RAN, e.g., circuitry power or cooling resources. Throughout this simulation, we set $P^{\mathrm{cir}} = \SI{38}{dBm}$. Equation \eqref{eq:phi} is a fraction of sum rate over the total network's power consumption.
	Note that such mathematical expression is widely used throughout the literature, e.g., \cite{8846706,9145076}.
	By varying the fronthaul capacity $C_n^\mathrm{max}$, the EE of different schemes is shown in Fig. \ref{xEE_yFthl_v3}.
	As $C_n^\mathrm{max}$ is low, we observe a significant EE gain from using RSMA over TIN and SCM-RSMA. We note that the one-layer based approach, i.e., SCM-RSMA, also reveals EE improvements over TIN, which highlights the advantages of \emph{rate-splitting}. In Fig. \ref{xEE_yFthl_v3}, the schemes converge to the same EE value at $C_n^\mathrm{max} = \SI{42}{Mbps}$. At this point, all schemes are able to fulfill the QoS demand in an efficient manner with low transmit powers, there is no gain from utilizing RSMA. However, as the EE of RSMA already saturates at $21$ Mbps, the proposed scheme is able to reach such point much faster than the reference schemes. In more details, Algorithm \ref{alg} determines a reasonable rate gap and transmit power trade-off, to achieve a good EE.
	
	Similarly, Fig. \ref{yEE_xSNR_v1} shows the EE as a function of the signal to noise ratio (SNR) in dB for $C_n^\mathrm{max} = \SI{21}{Mbps}$. In the high SNR regime, where interference becomes the major bottleneck in C-RAN, RSMA shows superior EE performance. This emphasizes the need for sophisticated interference management techniques when SNR $> \SI{2}{dB}$.
	
	In summary, these results emphasize the benefits of the RSMA paradigm in terms of EE enhancements, especially in high SNR range and limited fronthaul regime, where fulfilling all QoS demands is difficult.

	\subsection{Impact of QoS Targets}

	\begin{figure}[t]
	\centering
	\begin{subfigure}[t]{0.24\textwidth}
		\begin{center}
			\includegraphics[scale=0.59]{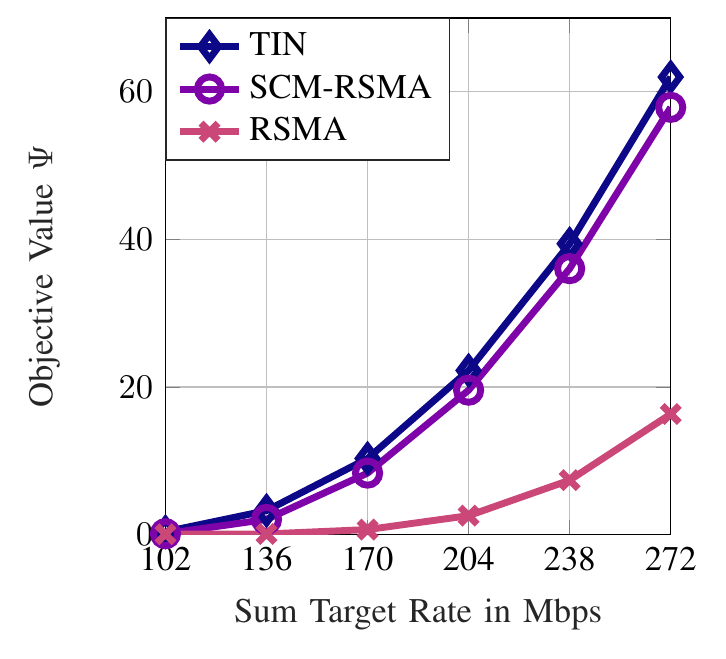}\vspace{-.2cm}
			\caption{$\Psi$ vs. sum target rate.}
			\label{yEE_xQOS_v2}
		\end{center}
	\end{subfigure}\hfill
	\begin{subfigure}[t]{0.24\textwidth}
		\centering
		\includegraphics[scale=0.59]{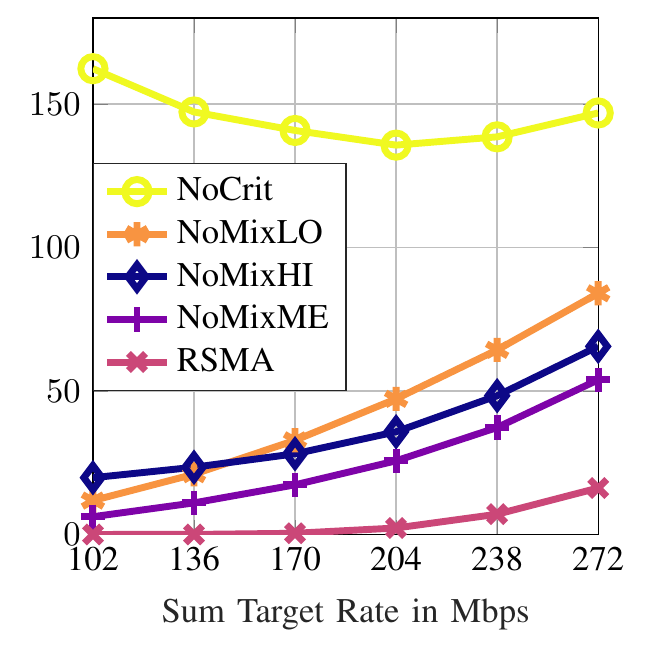}\vspace{-.2cm}
		\caption{$\Psi$ vs. sum target rate.}
		\label{yOBJ_xQOS_mix_v1}
	\end{subfigure}
	\caption{Objective function over sum target rate comparing various schemes.} \label{bla2}
    \end{figure}
    Since the exact values of QoS targets have significant impact on the system performance, we herein assess the objective value as a function of the sum target rate, i.e., $\sum_{k\in\mathcal{K}}r_k^\mathrm{des}$. Note that hereby we respect the criticality levels, i.e., we start with $12$, $6$, $\SI{3}{Mbps}$ for HI, ME, and LO levels, respectively. For each point on the x-axis, we increment the LO level by $\SI{1}{Mbps}$ keeping the ME (HI) level at double (quadruple) the LO level. We again compare RSMA, SCM-RSMA, and TIN, and observe the algorithms behavior in Fig. \ref{yEE_xQOS_v2}. In \eqref{eq:f}, the objective value $\Psi$ contains a summation of rate gap and total transmit power. Therefore, $\Psi$ increases jointly with higher sum target rates for all schemes. While TIN and SCM-RSMA experience a strong almost exponential incline, the proposed RSMA scheme keeps $\Psi$ much lower for increasing QoS demands. Especially at $\SI{272}{Mbps}$, the references achieve $\Psi=60$, while the RSMA objective is $18$.
    
    To show the impact of mixed criticality, in Fig. \ref{yOBJ_xQOS_mix_v1}, we show $\Psi$ as a function of the sum target rate. We compare the proposed method with three schemes ignoring the different criticality levels. Thereby, all QoS demands set to the same value, either the LO, ME, or HI level, which is denoted by NoMixLO, NoMixME, NoMixHI, respectively. Additionally, we consider a scheme discarding the QoS demands and thus the mixed criticality aspect as NoCrit. While the proposed scheme achieves best $\Psi$, NoCrit is not able to support the network needs, as is results in high values for $\Psi$, which constitutes bad EE and rate gap performance. NoMixLO, NoMixME, and NoMixHI exhibit lower $\Psi$ values, whereas NoMixME performs best. This is reasonable since NoMixME is a plausible compromise between setting all criticalities to LO or HI. However, these schemes do not beat the RSMA performance. Thus, these results show the importance of not only considering the criticality levels of the network, but also the mixed criticality aspect.
    
    As $\Psi$ captures a trade-off between minimizing the rate gap and the transmit power, lower values correspond to higher EE values, which was also verified by previous results. Thereby these simulations highlight the explicit gain of mixed criticality-enabled RSMA paradigm in terms of minimizing rate gap and transmit power, as well as achieving superior EE values in networks with high QoS demand, i.e., 6G networks.
    
    \subsection{Rate MSE vs. Transmit Power Minimization}\vspace*{-0.1cm}
    \begin{figure}[t]
	\centering
	\includegraphics[scale=0.66]{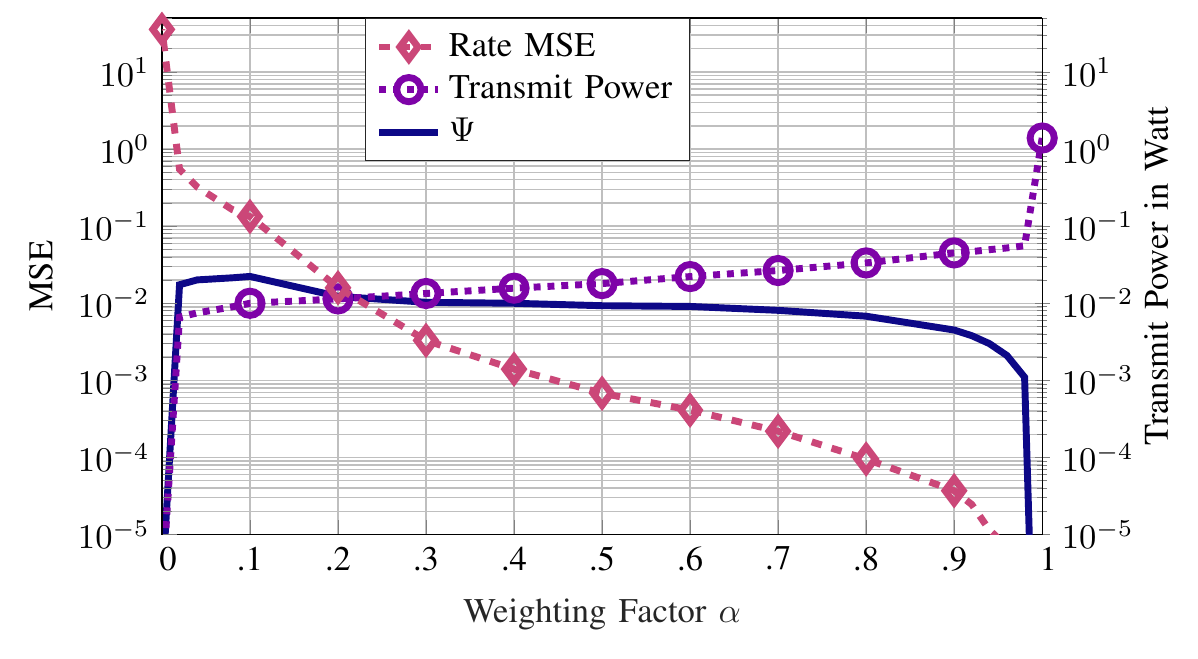}\vspace{-.2cm}
	\caption{Rate MSE and transmit power over $\alpha$.} \label{yMSE_yPWR_xAlph_v4}
    \end{figure}
	The manifold impacts of varying the parameter $\alpha$ in \eqref{eq:f} are now analyzed. Fig. \ref{yMSE_yPWR_xAlph_v4} shows the rate gap and transmit power as functions of $\alpha$. Additionally, we plot $\Psi$, which is the addition of both graphs with corresponding weights. 
	As expected, the numerical results show that there is a trade-off between optimizing the rate gap and the transmit power. This behavior comes from \eqref{eq:f}, as a large $\alpha$ will accentuate MSE minimization, while small $\alpha$ emphasizes transmit power minimization. In a wide-range of $\alpha \in \{0.3,\ldots,0.9\}$, the algorithm is able to achieve significant MSE improvements with only minor transmit power increase. As such, Fig. \ref{yMSE_yPWR_xAlph_v4} highlights the strong trade-off capabilities of the proposed scheme, which make is applicable in a wide-range of networks.
	
	In this context, the average number of iterations until convergence were $3$ when $\alpha=0.5$. Thereby, we validate the numerical merits of the proposed algorithm, especially regarding convergence and thus computation time. Such result is important in the context of latency sensitive applications in green 6G networks.\vspace*{-0.17cm}
	\section{Conclusion}\label{sec:con}\vspace*{-.1cm}
	Future 6G networks are envisioned to fulfill unprecedented QoS demands while at the same time minimize the power usage. As such, this paper addresses the problem of finding a reasonable trade-off between fulfilling the mixed-critical network-wide QoS target rates while also providing enhanced EE performance. Utilizing the benefits of C-RAN, the RSMA paradigm, and recent advances in fractional programming, this paper solves a non-convex rate gap and power minimization problem subject to various network constraints. An efficient iterative algorithm relying on QT is proposed, which is shown to provide a stationary solution to the original problem. The numerical results verify the EE benefits of RSMA over reference schemes such as SCM-RSMA and TIN in regards to fronthaul capacity, SNR, and QoS demands. Additionally, results emphasize the need for considering the aspects of mixed criticality. 
	The proposed scheme utilizes the benefits of RSMA and provides enhanced EE performance under excellent convergence behavior and is well suited for the future carbon neutrality and QoS demands for 6G.\vspace*{-0.17cm}
	
	
    \appendices
    \section{}\label{app2}\vspace*{-.1cm}
    With given Lemma \ref{lma_1}, the proof remains to show that the iterative procedure yields $(\bm{w}^\star,\bm{r}^\star,\bm{\gamma}^\star)$, a stationary solution to problem \eqref{eq:Opt2}. The steps in this proof are similar to the proof of \cite[Theorem 3]{8314727} and rely on the therein provided conditions of equivalent solution and equivalent objective. Algorithm \ref{alg} is a block coordinate descent algorithm for problem \eqref{eq:Opt3}, which is a convex problem. Therefore, such algorithm is guaranteed to converge to a stationary solution of problem \eqref{eq:Opt3}. Due to the definition of the auxiliary variables, equations \eqref{eq:qtp} and \eqref{eq:qtc} yield \eqref{eq:sinrp} and \eqref{eq:sinrc}, respectively, if and only if $({u}_{k}^p)^\star$ and $({u}_{i,k}^c)^\star$ are the optimal auxiliary variables of the stationary solution to \eqref{eq:Opt3}. As problems \eqref{eq:Opt2} and \eqref{eq:Opt3} share the same objective, and constraints \eqref{eq:qtp} and \eqref{eq:qtc} yield \eqref{eq:sinrp} and \eqref{eq:sinrc}, respectively, Algorithm \ref{alg} also converges to a stationary solution to problem \eqref{eq:Opt2}. This completes the proof.\vspace*{-0.2cm}
    \bibliographystyle{IEEEtran}
	\bibliography{bibliography}
\end{document}